\documentclass[12pt,draft,onecolumn]{IEEEtran}

\usepackage{amssymb}
\usepackage{amsmath}
\usepackage{amsthm}
\usepackage{epsfig}
\usepackage{mathrsfs}

\newcommand{\N}{\mathbb{N}}    %natural numbers
\newcommand{\R}{\mathbb{R}}    %real number field
\newcommand{\la}{\langle}
\newcommand{\ra}{\rangle}

\newcommand{\be}{\begin{equation}}
\newcommand{\ee}{\end{equation}}
\newcommand{\ba}{\begin{align}}
\newcommand{\ea}{\end{align}}
\newcommand{\bea}{\begin{eqnarray*}}
\newcommand{\eea}{\end{eqnarray*}}
\newtheorem{thm}{Theorem}[section]
\newtheorem{lem}[thm]{Lemma}

\newtheorem{rem}[thm]{Remark}

\begin{document}
\title{
A Sharp Restricted Isometry Constant Bound of Orthogonal Matching Pursuit
}
\author{
Qun Mo
\thanks
{
Research supported in part by the NSF of China under grant
10971189 and 11271010, and by the fundamental research funds
for the Central Universities.
%% the China Postdoctoral Science Foundation under grant 20100481430,
%% the Doctoral Program Foundation of Ministry of Education of China
%% under grant 20070335176 and Science Foundation of Chinese University
%% under grant 2010QNA3018.
}
\thanks
{
Q. Mo is with the Department of Mathematics, Zhejiang University,
Hangzhou, 310027, China (e-mail: moqun@zju.edu.cn ).
}
}

\maketitle

\begin{abstract}
We shall show that
if the restricted isometry constant (RIC) $\delta_{s+1}(A)$
of the measurement matrix $A$ satisfies
$$
\delta_{s+1}(A) < \frac{1}{\sqrt{s + 1}},
$$
then the greedy algorithm Orthogonal Matching Pursuit(OMP)
will succeed. That is, OMP
can recover every $s$-sparse signal $x$ in $s$ iterations from
$b = Ax$. Moreover, we shall show the upper bound of RIC is sharp
in the following sense. For any given $s \in \N$, we shall construct
a matrix $A$ with the RIC
$$
\delta_{s+1}(A) = \frac{1}{\sqrt{s + 1}}
$$
such that OMP may not recover some $s$-sparse signal $x$ in $s$
iterations.
\end{abstract}


\IEEEpeerreviewmaketitle

\section{Introduction}
\IEEEPARstart{W}{e} shall give a sharp RIC bound of OMP in this paper.
First let us review some basic concepts. Suppose we have a signal $x$
in $\R^N$ where $N$ is very large. For instance, $x$ represents a digit
photo of $1024\times1024$ pixels and $N = 1024^2 \approx 10^6$.
Denote $\|x\|_0$ to be the number of nonzero entries of $x$. We say $x$
is $s$-sparse if $\|x\|_0 \leq s$; and $x$ is sparse if $s \ll N$. In real
world, there are many signals are sparse
or can be well approximated by sparse signals, either under the canonical basis
or other special basis/frames. For simplicity, we assume that
$x$ is sparse under the canonical basis in $\R^N$.

Suppose we have some linear measurements of a signal $x$. That is, we have
$b = Ax$ where $A \in \R^{m \times N}$ and $b \in \R^m$ are given.
The key idea of compressed sensing \cite{CRT,D} is that
it is highly possible to retrieve
the sparse signal $x$ from those linear measurements,
while the number of measurements is far less than the dimension of the
signal, i.e., $m \ll N$.

%% Usually we hope that the measurement matrix $A$
%% satisfies some properties so that we can
%% successfully retrieve the sparse signal $x$ from
%% those linear measurements $b$.
To retrieve such a sparse signal $x$, a natural method is to solve the
following $l_0$ problem
 \be\label{p0}
 \min_{x} \|x\|_0 \quad \text{subject to} \quad Ax = b
 \ee
where $A$ and $b$ are known. To ensure the $s$--sparse solution is
unique, we would like to use the \emph{restricted isometry
property} (RIP) which was introduced by Cand\`es and Tao in
\cite{CT}. A matrix $A$ satisfies the RIP of order $s$ with the
\emph{restricted isometry constant} (RIC) $\delta_s = \delta_s(A)$
if $\delta_s$ is the smallest constant such that
 \be\label{RIP}
 (1-\delta_s)\|x\|_2^2
 \leq
 \|A x \|_2^2
 \leq
 (1+\delta_s)\|x\|_2^2
 \ee
holds for all $s$-sparse signal $x$.

If $\delta_{2s}(A) < 1$, the $l_0$ problem has a unique $s$-sparse
solution \cite{CT}. The $l_0$
problem is equivalent to the $l_1$ minimization problem when
$\delta_{2s}(A) < \sqrt{2}/2$,
please see \cite{C, ML, CZ} and the references therein.

OMP is an efficient greedy algorithm to solve the
$l_0$ problem \cite{DMA, PRK, T}. To introduce this algorithm, we
denote $A_{\Omega}$ by the matrix $A$ with indices of its columns in $\Omega$
and we denote $x_{\Omega}$ as the similar restriction of some vector $x$.
The following iterative algorithm shows the framework of OMP.

\begin{center}
\begin{itemize}
\item [] \textbf{Input}: $A$, $b$
\item [] \textbf{Set}: $\Omega_0=\emptyset$, $r_0 = b$, $k=1$
\item [] \textbf{while not converge}
\begin{itemize}
\item $\Omega_{k}= \Omega_{k-1} \cup \arg\max_{i} |\langle
r_{k-1}, A e_i \rangle|$
\item
$x_k = \arg\min_z \| A_{\Omega_k} z - b \|_2$
\item
$r_k = b - A_{\Omega_k} x_k$
\item
$k = k+1$
\end{itemize}
\item [] \textbf{end while}
\item [] $\hat{x}_{\Omega_k} = x_k$,
$\hat{x}_{\Omega_k^C} = 0$
\item [] \textbf{Return} $\hat{x}$
\end{itemize}
\end{center}

%------------------------------------------------------------------------------------------
Now let us consider the conditions in terms of RIC for OMP to
exactly recover any $s$-sparse signal in $s$ iterations.
Davenport and Wakin \cite{DW} have proven that
$\delta_{s+1}(A) < {1}/{(3\sqrt{s})}$ is sufficient. Later, Liu and
Temlyakov \cite{LT} have improved the condition to
$\delta_{s+1}(A) < {1}/{((\sqrt{2} + 1)\sqrt{s} \,)}$.
Furthermore, Mo and Shen \cite{MS} have pushed the sufficient condition to
$\delta_{s+1}(A) < {1}/{(\sqrt{s} + 1)}$ and
have shown by detailed examples that a necessary condition
is $\delta_{s+1}(A) < {1}/{\sqrt{s}}$.
Later, Yang and Hoog \cite{YH} improved the
sufficient condition to
$ \delta_{s}(A) + \sqrt{s} \, \delta_{s+1}(A) < 1$.

There is a small gap open between the sufficient condition and the necessary
condition of the best results. And we would like to fill out the gap
in this paper.

%=====================================================================
The content of this paper is consisted by two parts.
\begin{itemize}
\item We shall prove that the condition
$$
\delta_{s+1}(A) < \tfrac{1}{ \sqrt{s + 1} }
$$
is sufficient for OMP to exactly recover any $s$-sparse $x$ in $s$
iterations.

\item For any positive integer $s$, we shall construct a matrix with
$\delta_{s+1}={1}/{ \sqrt{s + 1}}$ such that OMP may fail for at least
one $s$-sparse signal in $s$ iterations.

\end{itemize}

\section{Preliminaries}\label{section1}
Before going further, let us introduce some notations. For
$k = 1, 2, \cdots, N$, define
$e_k$ to be the $k$-th element in the canonical basis of $\R^N$.
That is, $e_k$ is the 1-sparse vector in $\R^N$ such that only the $k$-th
entry of $e_k$ is 1. For a given matrix $A$
and a given signal $x$, define
$$
S_k := \langle Ax, A e_k \rangle, \quad k=1,\ldots, N.
$$
Denote $ S_0: = \max_{k\in\{1,\ldots, s\}} |S_k| $. The following
two lemmas are useful in our analysis.

\begin{lem}\label{lem1}
For any $s > 0$, define $t := \pm (\sqrt{s + 1} - 1)/\sqrt{s}$. Then we have
$t^2 < 1$ and
\be
\label{eq_sqr}
\begin{aligned}
& \| A(x + t e_k) \|_2^2
- \| A(t^2 x - t e_k) \|_2^2
= \\
& (1 - t^4) \left( \la Ax, Ax \ra \pm \sqrt{s} \la Ax, A e_k \ra \right)
\;\; \forall k = 1, \cdots, N.
\end{aligned}
\ee
\end{lem}

\begin{proof}
By the definition of $t$, $t^2 = (\sqrt{s + 1} - 1)^2/s =
(\sqrt{s + 1} - 1)/(\sqrt{s + 1} + 1) < 1$.
Moreover, by direct expanding and simplification, we have
$$
\| A(x + t e_k) \|_2^2 - \| A(t^2 x - t e_k) \|_2^2
= (1 - t^4) \left( \la Ax, Ax \ra + \dfrac{2t}{1 - t^2} \la Ax, A e_k \ra \right)
.
$$
Now we only need to verify that ${2t}/({1 - t^2}) = \pm \sqrt{s}$.
This is indeed true since
$$
\dfrac{2t}{1 - t^2} = \pm 2 \dfrac{\sqrt{s + 1} - 1}{\sqrt{s}} /
\left(1 - \dfrac{\sqrt{s + 1} - 1}{\sqrt{s + 1} + 1} \right)
= \pm \sqrt{s}
.
$$
\end{proof}

\begin{lem}\label{lem2}
If the restricted isometry constant $\delta_{s+1}(A)$ satisfies
 \be\label{delta}
 \delta_{s+1}(A) < \frac{1}{ \sqrt{s + 1} }
 \ee
and the support of the signal $x$ is a non-empty subset of
$\{1,2,\ldots,s\}$, then $S_0> |S_k|$ for $k > s$.
\end{lem}

\begin{proof}
Notice this lemma keeps unchanged if we replace $x$ by $cx$ with
some non-zero scalar $c$. Therefore, we can assume that $\| x \|_2 = 1$.
For the given $s$-sparse $x$, we obtain
$$
\begin{aligned}
 & \langle Ax,Ax \rangle
 = \left \la A \sum_{k=1}^s{x_k e_k},Ax \right \ra
\\ & = \sum_{k=1}^s x_k \la A e_k, Ax\ra
\\ & = \sum_{k=1}^s x_k S_k
%% \eea
%% It follows
%% \bea
%% &\langle  Ax,Ax  \rangle \\
%% & = & \sum_{k=1}^s x_k S_k \\
\\ & \leq S_0 \|x\|_1
\\ & \leq S_0 \sqrt{s} \|x\|_2
\\ & = \sqrt{s} S_0.
\end{aligned}
$$
Define $t := -(\sqrt{s + 1} - 1)/\sqrt{s}$. Then by the above
inequality and Lemma~\ref{lem1}, we have
%% This implies
\be
\label{eq2}
\begin{aligned}
& (1 - t^4) \sqrt{s}(S_0 - S_k)
\geq (1 - t^4) (\la Ax,Ax \ra - \sqrt{s} \la Ax, A e_k \ra)
\\
= & \| A(x + t e_k) \|_2^2 - \| A(t^2 x - t e_k) \|_2^2.
\end{aligned}
\ee
Moreover, for any $k > s$, by the definition of $\delta_{s + 1}(A)$,
we obtain
\be
\label{eq3}
\begin{aligned}
& \| A(x + t e_k) \|_2^2 - \| A(t^2 x - t e_k) \|_2^2
\\ & \geq \left( 1 - \delta_{s + 1}(A) \right) (\|x\|_2^2 + t^2)
- t^2 \left( 1 + \delta_{s + 1}(A) \right)(t^2 \|x\|_2^2 + 1)
\\ & = (1 + t^2)\left(1 - t^2 - (1 + t^2) \delta_{s + 1}(A) \right)
\\ & = (1 + t^2)^2 \left( (1 - t^2)/(1 + t^2) - \delta_{s + 1}(A) \right)
.
\end{aligned}
\ee
By the definition of $t$, we have
\be
\label{eq4}
\begin{aligned}
(1 - t^2)/(1 + t^2) =
\left( 1 - \dfrac{\sqrt{s + 1} - 1}{\sqrt{s + 1} + 1} \right)/
\left( 1 + \dfrac{\sqrt{s + 1} - 1}{\sqrt{s + 1} + 1} \right)
= 1/\sqrt{s + 1}
.
\end{aligned}
\ee
Now combining (\ref{eq2}), (\ref{eq3}), (\ref{eq4})
and condition (\ref{delta}), we obtain
$$
\begin{aligned}
 & (1 - t^4) \sqrt{s}(S_0 - S_k)
\\ & \geq \| A(x + t e_k) \|_2^2 - \| A(t^2 x - t e_k) \|_2^2
\\ & \geq (1 + t^2)^2 \left( (1 - t^2)/(1 + t^2) - \delta_{s + 1}(A) \right)
\\ & = (1 + t^2)^2 \left( 1/\sqrt{s + 1} - \delta_{s + 1}(A) \right)
\\ & > 0
.
\end{aligned}
$$
Therefore, we have $S_0 > S_k$ for all $k > s$.
Similarly we can prove that $S_0 > -S_k$ for all $k > s$.
Hence, $S_0 > |S_k|$ for all $k > s$.
\end{proof}

\section{Main Results}
Now we are ready to show the main results of this paper.

\begin{thm}\label{thm1}
Suppose that $A$ satisfies
%% the RIP of order $K+1$ with the
%% restricted isometry constant
 \be\label{key}
 \delta_{s+1}(A) < \frac{1}{\sqrt{s+1}},
 \ee
then for any $s$-sparse signal $x$, OMP will recover $x$ from
$b = Ax$ in $s$ iterations.
\end{thm}
\begin{proof}
Without loss of generality, we assume that the support of $x$ is a
subset of $\{1,2,\ldots, s\}$. Thus, the sufficient condition for OMP
choosing an index from $\{1,\ldots, s\}$ in the first iteration is
$$
S_0> |S_k |\ \text{for all} \ k > s.
$$
By Lemma \ref{lem2}, $\delta_{s+1} < \frac{1}{\sqrt{s+1}}$
guarantees the success of the first iteration of OMP. OMP makes an
orthogonal projection in each iteration. Hence, it can be proved by
induction that OMP always selects an index from the support of $x$
in $s$-iterations.
\end{proof}

\begin{thm}\label{thm2}
For any given positive integer $s$, there exist a $s$-sparse
signal $x$ and a matrix $A$ with the restricted isometry constant
$$
\delta_{s+1}(A) = \frac{1}{\sqrt{s + 1}}
$$
such that OMP may fail in $s$ iterations.
\end{thm}

\begin{proof}
For any given positive integer $s$, let
$$
 A = \left(\begin{array}{cccc}
 & & & \frac{1}{\sqrt{s(s + 1)} } \\
 & \sqrt{\frac{s}{s + 1}} I_{s} & & \vdots \\
 & & & \frac{1}{\sqrt{s(s + 1)}} \\
 0 & \ldots& 0 & 1
\end{array}
\right)_{(s+1)\times (s+1)}.
$$
By simple calculation, we get
$$
 A^T A = \left(\begin{array}{cccc}
 & & & \frac{1}{s + 1} \\
 & \frac{s}{s + 1}I_{s}& & \vdots \\
 & & & \frac{1}{s + 1} \\
\frac{1}{s + 1} & \ldots& \frac{1}{s + 1} & 1 + \frac{1}{s + 1}
\end{array}
\right)_{(s+1)\times (s+1)},
$$
where $A^T$ denotes the transpose of $A$.
By direct calculation, we can verify that
$A^T A x = s x/(s + 1)$ for all $x \in \R^{s + 1}$ satisfying
$x_{s + 1} = 0$ and $\sum_{k = 1}^s x_k = 0$.
Therefore, $s/(s + 1)$ is an eigenvalue of $A^T A$ with multiplicity of $s - 1$.
Moreover, by direct calculation, we can see that
$A^T A x = (1 \pm 1/\sqrt{s + 1}) x$ with
$x = (1, 1, \ldots, 1, 1 \pm \sqrt{s + 1})^T$.
Therefore, $A^T A$ has another two eigenvalues $1 \pm 1/\sqrt{s + 1}$.
Thus, the restricted isometry constant $\delta_{s+1}(A)$ is
equal to $\tfrac{1}{\sqrt{s + 1}}$.
Now let
$$
x = (1,1,\ldots,1,0)^T\in \mathbb R^{s+1}.
$$
We have
$$
S_k = \langle Ae_k, Ax \rangle = \langle A^TAe_k, x \rangle
=s/(s + 1)\
\text{for all}\ k \in \{1, \ldots, s + 1\}.
$$
This implies OMP may fail in the first iteration. Since OMP chooses
one index in each iteration, we conclude that OMP may fail in $s$
iterations for the given matrix $A$ and the $s$-sparse signal $x$.
\end{proof}

\begin{rem}
For the case of measurements with noise, please see \cite{SL}, \cite{Cdr},
and the references therein.
%% It is challenging to design a measurement matrix having a very small
%% restricted isometry constant $\delta_{K+1}$; and Theorem \ref{thm2}
%% shows that this kind of requirement is necessary. However, if we
%% select multiple indices per iteration, we can recover the
%% $K$--sparse signal given in Theorem \ref{thm2} in $K$ iterations.
%% Actually, this technique has been widely used in many related greedy
%% pursuit algorithms, such as CoSaMP \cite{NT} and Subspace Pursuit
%% algorithm \cite{DM}.
\end{rem}

\begin{biographynophoto}
%[{\includegraphics[width=1in,height=1.25in,clip,keepaspectratio]{123.jpg}}]
{Qun Mo} was born in 1971 in China. He has obtained a B.Sc. degree in 1994
from Tsinghua University, a M.Sc. degree in 1997 from Chinese Academy of
Sciences and a Ph.D. degree in 2003 from University of Alberta in Canada.

He is current an associate professor in mathematics in Zhejiang University.
His research interests include compressed sensing, wavelets and their applications.
\end{biographynophoto}

\end{document}